\documentclass[journal]{IEEEtran}
\makeatletter
\def\ps@headings{%
\def\@oddhead{\mbox{}\scriptsize\rightmark \hfil \thepage}%
\def\@evenhead{\scriptsize\thepage \hfil\leftmark\mbox{}}%
\def\@oddfoot{}%
\def\@evenfoot{}}
\makeatother
\usepackage{amssymb,amsmath}
\usepackage{paralist,color,clock}
\usepackage{algorithmic,algorithm2e}
\usepackage{graphicx,cite}

\newcommand{\N}{\mathcal{N}}

\newtheorem{theorem}{\textbf{Theorem}}
\newtheorem{lemma}[theorem]{\textbf{Lemma}}
\newtheorem{definition}[theorem]{\textbf{Definition}}

\newcommand{\nix}[1]{}

\begin{document}
\title{\LARGE{Distributed Data Collection and Storage Systems for Collaborative Learning Vision
Sensor Devices with Applications to Pilgrimage}}
%%%%%%%%%%%%%%%%%%%

\author{Salah A. Aly \\Department  of Computer Science\\ Center of Research Excellence in Hajj and Omrah, HajjCore \\ Umm Al-Qura University,  Makkah,  KSA \\Email: salahaly@uqu.edu.sa

\thanks{This research is funded by the Center of Research Excellence in Hajj and Omrah at UQU, Makkah, KSA.}
}

 \maketitle

\begin{abstract}
This work presents novel distributed data collection systems and storage algorithms for collaborative learning wireless sensor networks (WSNs).
In a large  WSN, consider $n$ collaborative sensor  devices distributed randomly to acquire information and learn about a certain field. Such sensors have less power, small bandwidth, and short memory, and they might disappear from the network  after certain time of operations. The goal of this work  is to design efficient  strategies to learn about the field by collecting sensed data from these $n$ sensors with less computational overhead and efficient  storage encoding operations.

In this data collection system, we propose two distributed data storage algorithms (DSA's) to solve this problem with the means of network flooding and connectivity among sensor devices. In the first algorithm denoted, DSA-I, it's  assumed that the total number of nodes is known for each node in the network. We show that this algorithm is efficient in terms of the encoding/decoding operations. Furthermore, every node uses network flooding to disseminate its data throughout the network using mixing time approximately $O(n)$. In the second algorithm denoted, DSA-II, it's  assumed that the total number of nodes is not known for each learning sensor, hence dissemination of the data does not depend on the value of $n$. In this case we show that the encoding operations take $O(C\mu^2)$, where $\mu$ is the mean degree of the network graph and  $C$ is a system parameter. Performance of these two algorithms match the  derived theoretical results. Finally, we show how to deploy these algorithms for monitoring and measuring certain phenomenons in American-made camp tents located in Minna field in south-east side of Makkah.
\end{abstract}

%%%%%%%%%%%%%%%%%%%%%%%%%%%%%%%%%%%%%%%%%%%%%%%%%%%%%%%%%%%%%%%%%%%%%%%%
\section{Introduction}\label{sec:intro}
 The field of information technology has    witnessed remarkable extensions especially after appearance of the world wide web two decades ago. In addition, this has been embarked by appearance of several communication networking branches, such as wireless sensor networks.
Wireless sensor networks (WSN's) consist of small devices (nodes)  with low CPU power,
small bandwidth, and limited memory. They can be deployed in isolated, tragedy, and obscured fields to monitor objects,
detect fires or floods, measure temperatures, transmit  media streams, and etc. They can also be used in areas  where human involvement is difficult to reach or it is danger for human being to be involved.
There has been extensive research work  on sensor networks to improve their
services, powers, and operations~\cite{stojmenovic05}.  They have taken much attention recently due to their varieties of  applications. Much research has been done in both academia and industry to increase their reliability, usage, and operations.

We consider  a model for large-scale wireless sensor networks where $n$ data collection and storage sensor  nodes are distributed uniformly and randomly. These $n$ nodes are deployed to
collect information  and transmit media streams (images, videos, texts) about a certain field. These $n$  sensor devices have a
short time-to-live, limited memory,  and might disappear from the network at
anytime. Also, the nodes do not know locations of the neighboring
nodes, and they do not maintain routing tables to forward messages. We assume
that the $n$ sensing and data collection nodes generate independent packets that can be classified
as initial or update packets sent at an arbitrary time. A packet  initiated from a node $u$ contains its $ID_u$, time-to-live parameter, and sensed data. In addition, ever storage node $u$ has $M$ buffer size that can be divided into $m$ small buffers to save other neighbors' data. Every storage node decides randomly and independently from which it will accept or reject packets. Also, a packet will be discarded once it travels through the network $O(n)$.

\begin{figure}
  \begin{center}
  \includegraphics[width=8.8cm,height=5.2cm]{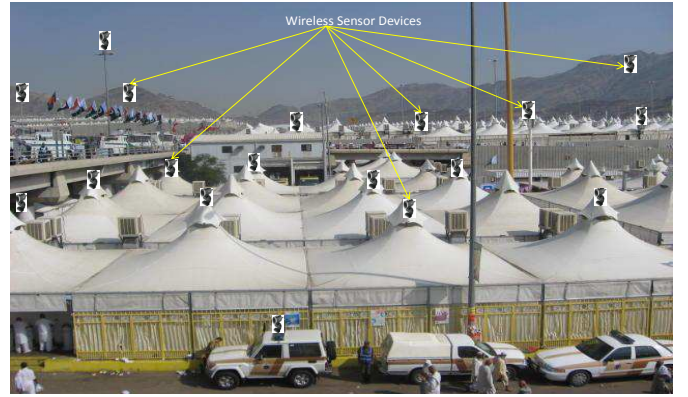}

  \caption{A wireless sensor network consists of various small devices with limited CPU power, small memory and bandwidth. Collaborative Sensor nodes are distributed randomly to monitor, collect data,  and learn about Minna field in the east of Makkah. Approximately 50.000 camp tents are located in Minna to accommodate 3-5 million people for 4-8 days during pilgrimage, according to  2010 KSA statistics.}\label{fig:field1}
  \end{center}
\end{figure}

The goal of this work is to develop an efficient method to randomly distribute and collect  information from $n$ sensors to all $n$ storage nodes. In this case, a data collector with a high computational power can query any $(1+\epsilon)n/m$ storage nodes for $\epsilon >0$, and easily
retrieve information about the $n$ sensor nodes with a high probability. Other versions of this problem has been solved by using coding in a centralized way (e.g. Fountain
codes, MDS and linear codes) by adding some redundancy, where a node can send its data to a pre-selected set of other nodes in the network~\cite{dimakis06b,dimakis08,lin07a,aly11a}. Over a distributed random network,  this
is unreliable since we still need to find a strategy to distribute the
information from the sources to a set of arbitrary storage nodes. Hence, a
decentralized way solution is needed where the data collector and storage nodes are distributed randomly
and independently. Therefore, the considered problem  is a network storage problem rather
than a network transmission problem. The later problem assumes that channel coding and modulation theory are used to handle the transmission for a source to a destination.  The former problem requires  distributed networking storage algorithms to assure protection of information against node failures or disappearance. It is assumed that all nodes trust each other data, and attackers are unable to break the nodes transitions.

\begin{figure}
  \begin{center}
  \includegraphics[width=8.8cm,height=5.2cm]{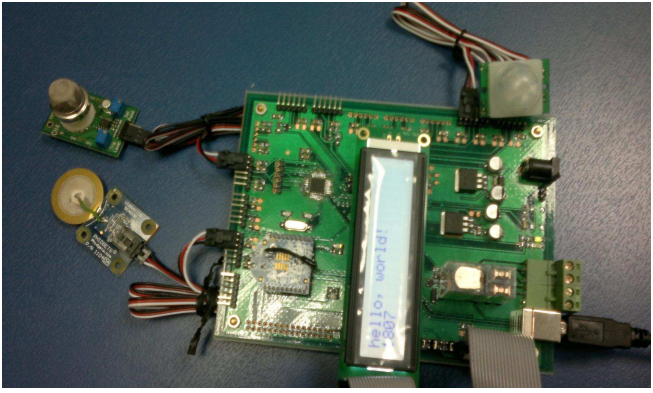}

  \caption{A wireless sensor devices equipped with several sensor components to measure temperature, gas, pollution, and co2.}\label{fig:sensors}
  \end{center}
\end{figure}

The motivations for this work are that:
\begin{compactenum}[i)]
\item
We demonstrate a realistic model for WSN's, where nodes are distributed randomly with limited power and memory.
\item The encoding and decoding operations are done linearly.
\item
Querying only $(1+\epsilon)n/m$ a subset of the network reveals information about all nodes.

\item
The proposed storage algorithms have less computational complexity in comparison to the related work shown in Section~\ref{sec:relatedwork}.
\end{compactenum}

This work is organized as follows. In Section~\ref{sec:relatedwork} we present a background and  short survey of the related work. In Section~\ref{sec:model} we
introduce the network model. In Sections~\ref{sec:algDSA-I} and~\ref{sec:algDSA-II} we propose two storage algorithms and provide their analysis in Sections~\ref{sec:analysis} and~\ref{sec:analysis2}, respectively. In
Section~\ref{sec:simulation} we present simulation studies of the proposed
algorithms, and the work is concluded in Section~\ref{sec:conclusion}.

\begin{figure}
\begin{center}
\includegraphics[scale=0.7]{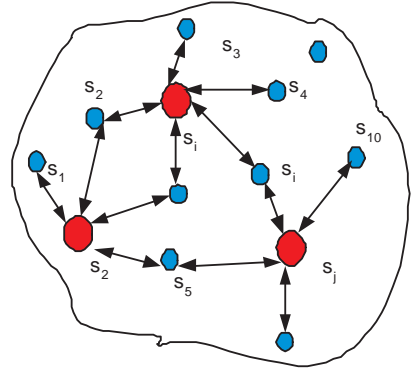}
\caption{A WSN with $n$ nodes arbitrary and randomly distributed in a field. A node $s_i$ determines its degree $d(s_i)$ by sending a flooding message to the neighboring nodes. }
\label{fig:wsn1}
\end{center}
\end{figure}
%%%%%%%%%%%%%%%%%%%%%%%%%%%%%%%%%%%%%%%%%%%%%%%%%%%%%%%%%%%%%%%%%%%%%%%%
\section{Network Model and Assumptions}\label{sec:model}

In this section we present the  network model and problem definition. Consider a wireless sensor network $\N$ with $n$ sensor nodes that are uniformly distributed at random in a region $\mathcal{A}=[L,L]^2$ for some integer $L\geq 1$. The network model $\N$  can be considered as an abstract graph $G=(V,E)$ with  a set of  nodes $V$ and a set of edges $E$. The set $V$ represents the sensors $S=\{s_1,s_2,\ldots,s_n\}$ that will measure information about a specific field. Also, $E$ represents a set of connections (links) between the sensors $S$. Two arbitrary sensors $s_i$ and $s_j$ are connected if they are in each other transmission range.

We ensure that the network is dense, meaning with high probability there are no isolated nodes. Let $r>0$ be a fraction, we say that two nodes $u$ and $v$ in $V$ are connected in $G$ if and only if the distance between them is bounded by the design parameter $r$, i.e. $0<d(u,v) \leq r$. Put differently, let $z$ be a random variable represents existence of an edge between any two arbitrary nodes $u$ and $v$. Then

\begin{eqnarray}z=\left\{ \begin{array}{cl}
1 & \mbox{ $d(u,v) \leq r$ }\\
0 & otherwise\\
\end{array}\right.
\end{eqnarray}

One can guarantee such condition by assuming that the radius  $r \geq O(\frac{1}{n^2})$.

\subsection{Assumptions}

We have the following assumptions about the network model $\N$:
\begin{compactenum}[i)]
\item Let $S=\{s_1,\ldots,s_n\}$ be a set of sensing nodes that are
    distributed randomly and uniformly in a field. Also, they are the set
    of storage nodes. So, this assumption differentiate between our work
    and the problems considered in~\cite{aly08e,lin07a}.
\item Every node does not maintain routing or geographic tables, and the network topology is not known.  Every node $s_i$ can send
    a flooding message to the neighboring nodes. Also,
    every node $s_i$ can detect the total number of neighbors by sending a
    simple flooding query message, and whoever replies to this message
    will be a neighbor of this node. Therefore, our work is more general
    and different from the work done in~\cite{dimakis07,dimakis05}. The
    degree $d(u)$ of this node is the total number of neighbors with a
    direct connection.

\item  Every node has  a buffer of size $M$ and this buffer can be divided
    into smaller buffers, each of size $c$, such that $m=\lfloor M/c
    \rfloor$. Hence, all nodes have the same number of buffers. Also, the first buffer of a node $u$ is reserved for its own sensing data.

\item Every node $s_i$ prepares a packet $packet_{s_i}$ with its ID,
    sensed data $x_{s_i}$, counter $c(x_{s_i})$, and a flag that is set to zero or one.
\begin{eqnarray}
packet_{s_i}=(ID_{s_i},x_{s_i},c(x_i), flag)
\end{eqnarray}
The flag is set to zero when the sensors initiate data for the first time, otherwise it will be set to one for data update.
\item We will consider two different types of packets: initialization and
    update packets. One can consider these two cases by using a flag that
    takes the values zero and one. If the source node sends a packet and
    the flag is set to zero, then it will be considered as an
    initialization packet. Otherwise, it will be considered as an update
    packet. The packets sent from all sources at the beginning of sensing
    phase are considered initialization packets.

\item Every node draws  a degree $d_u$ from a degree distribution
    $\Omega$. If a node decided to accept a packet, it will also decide on
    which buffer it will be stored.

\end{compactenum}

When  a node $s_i$ receivers a packet, it will decide to either reject or accept it with  a certain probability.\\

\begin{figure}
\begin{center}
\scalebox{0.7}{\includegraphics[width=8cm,height=3.5cm]{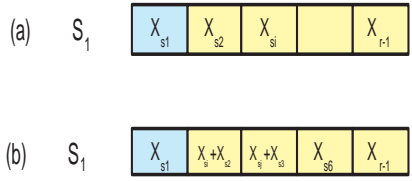}}
\caption{Every node $s_i$ has a buffer of size $M$ that  is divided into $m$ small buffers. The node $s_i$ decides  with a certain probability whether to accept or reject a data $x_{s_j}$ and where to save it in one of its buffers. }
\label{fig:wsn1}
\end{center}
\end{figure}

\bigskip
%%%%%%%%%%%%%%%%%%%%%%%%%%%%%%%%%%%%%%%%%%%%%%%%%%%%%%%%%%%%%%%%%%%%%%%%
\section{Distributed Storage Algorithms}\label{sec:algDSA-I}
In this section we will present a networked distributed storage algorithm
for wireless sensor networks and study its encoding and decoding operations. Other previous algorithms
assumed that $k$ source nodes disseminate
their sensed data throughout a network with $n$ storage nodes using the
means of Fountain codes and random walks. However, in this work we
generalize this scenario where  a set of $n$ sources disseminate their
data to a set of  $n$ storage nodes. Also,  in this proposed algorithm we
use properties of wireless sensor networks such as broadcasting and
flooding.

\subsection{Encoding Operations}
We present a distributed storage  algorithm (DSA-I) for wireless
sensor networks. DSA-I algorithm consists of three main steps: Initialization,
encoding/flooding, and storage phases. Each phase can be described as
follows.

\begin{enumerate}[I)]
\item \textbf{Initialization Phase:} Every node $s_i$ in $S$ has an
    $ID_{s_i}$ and reading (sensing) data $x_{s_i}$. The node $s_i$ in the
    initialization phase prepares a $packet_{s_i}$ along with its
    info, a counter $c(x_{s_i})$ that determines the maximum number of
    hops that will receive $x_{s_i}$, and a flag that is set to zero.
    We ensure that every message $x_{s_i}$ will have it is own
    threshold value $c(x_{s_i})$ set by the sender $s_i$ based on the
     set of neighbors $\N(s_i)$. This value will depend on the degree $d_(s_i)$. If the node $s_i$ has a few
    neighbors, then $c(x_{s_i})$ will be large. Also, a node with
    large number of neighbors will choose a small counter
    $c(x_{s_i})$. This means that every node will decide its own
    counter.
\begin{eqnarray}
packet_{s_i}=(ID_{s_i}, x_{s_i},c(x_{s_i}), flag)
\end{eqnarray}

The node $s_i$ broadcasts this packet to all neighboring nodes $\N(s_i)$.

\item \textbf{Encoding and Flooding Phase:}
\begin{itemize}
\item After the flooding phase, every node $u$ receiving the
    $packet_{s_i}$ will accept the data $x_{s_i}$ with probability
    one and will add this data to its buffer $y$.
\begin{eqnarray}
y_u^+=y_u^- \oplus x_{s_i}.
\end{eqnarray}
\item The node $u$ will decrease the counter by one as
    \begin{eqnarray}c(x_i)=c(x_i)-1.\end{eqnarray}
\item The node $u$ will select a set of neighbors that did not
    receiver the message $x_{s_i}$ and it will send this message
    using multicasting.
\item For an arbitrary node $v$ that receives the message from
    $u$, it will check if the $x_{s_i}$ has been received before,
    if yes, then it will discard it. If not, then it will run a
    probability distributed whether to accept or reject it. If
    accepted, then it will add the data to its buffer $
    y_v^+=y_v^- \oplus x_{s_i} $ and will decrease the counter
    $c(x_i)=c(x_i)-1.$

\item The node $v$ will check if the counter is zero, otherwise it will decrease it and send this message to the neighboring nodes that did not receive it using multicasting.

\end{itemize}

\item \textbf{Storage Phase:} Every node will maintain its own buffer
    by storing a copy of its data and other nodes' data. Also, a node
    will store a list of nodes ID's of the packets that reached it.
    After all nodes receive, send and storage their own and
    neighboring data, every node will be able to maintain a buffer
    with some data of the network nodes.
\end{enumerate}

\begin{algorithm}[t!]
\SetLine%
\KwIn{A sensor network with $S=\{s_1,\ldots,s_n\}$ source nodes,  $n$ source packets $x_{s_i},\ldots,x_{s_n}$ and a positive constant $c(s_i)$.}%
\KwOut{storage buffers $y_1,y_2,\ldots,y_n$ for all sensors $S$.} %
\ForEach{node $u=1:n$}%
    {
    Generate $d_c(u)$ according to $\Omega_{is}(d)$ (or $\Omega_{rs}(d)$ and a set of neighbors $\N(u)$ using flooding.;%
    }%
\ForEach{source node $s_i, i=1:n$}%
    {
    Generate header of $x_{s_i}$ and $token=0$\;%
    Set counter $c(x_{s_i})=\lfloor n/d(s_i)\rfloor$\;%
   Flood $x_{s_i}$ to all $\mathcal{N}(s_i)$ uniformly at random, Send $x_{s_i}$ to $u \in \mathcal{N}(s_i)$ \;%

 with probability 1,  $y_u$ = $y_u\oplus x_{s_i}$\;%
 Put $x_{s_i}$ into $u$'s forward queue\;%
    $c(x_{s_i})=c(x_{s_i})-1$\;%
    }%

\While{source packets remaining}%
    {
    \ForEach{node $u$ receives packets before current round}%
        {
        Choose $v\in \mathcal{N}(u)$ uniformly at random\;%
        Send  packet $x_{s_i}$ in $u$'s forward queue to $v$\;%
        \uIf{$v$ receives $x_{s_i}$ for the first time}%
            {
            coin = rand(1)\;
				{flip a coin to accept or reject a packet} \;

            \If{$\mbox{coin}\leq \frac{1}{d_c(v)}$}%
                {
                $y_v$ = $y_v\oplus x_{s_i}$\;%
                Put $x_{s_i}$ into $v$'s forward queue\;%
            $c(x_{s_i})=c(x_{s_i})-1$%
                }%
            }%
            \uElseIf{$c(x_{s_i})\geq 1$}%
                {
                Put $x_{s_i}$ into $v$'s forward queue\;%
                $c(x_{s_i})=c(x_{s_i})-1$\;%
                }
            \Else
                {
                Discard $x_{s_i}$\;
                Hence $C(s_i)=1$ or no node to send to.
                }%
            }%
        }%

\mbox{}\\%x

\caption{ DSA-I Algorithm: Distributed storage algorithm  for a WSN where
the data is disseminated using multicasting and flooding to all
neighbors.} \label{alg:DSA-I}
\end{algorithm}

\subsection{Decoding Operations}

The stored data can be recovered by querying a number of nodes from the network. Let $n$ be the total number of alive nodes; assume that every node
has $m$ buffers such that $\lfloor M/c \rfloor$ is the number of buffers,
where $c$ is  a small buffer size, and $M$ is total buffer size by a node . Then the
data collector needs to query at least $(1+\epsilon)n/m$ nodes in order to
retrieve the information about the $n$ variables.

\bigskip
%%%%%%%%%%%%%%%%%%%%%%%%%%%%%%%%%%%%%%%%%%%%%%%%%%%%%%%%%%%%%%%%%%%%%%%%
\section{DSA-I Analysis}\label{sec:analysis}
We shall provide analysis for the DSA-I algorithm shown in the previous section. The main idea is to utilize flooding and the node degree of each node to disseminate the sensed data from sensors throughout the network. We note that nodes with large degree will have  smaller counters in their packets such that their packets will travel for minimal number of neighbors. Also, nodes with smaller degree will have larger counters such that their packets will be disseminated to many neighbors as possible.

The following lemma establishes  the number of hobs (steps) that every packet will travel in the network.
\begin{lemma}\label{lem:onepacket}
On average with a high probability,  the total number of steps for one packet originated by a node $u$ in one branch in DSA-I is given by
\begin{eqnarray}
O(n/ \mu).
\end{eqnarray}

%\begin{eqnarray}
%c(x_u)=\lfloor n/d_n(u)\rfloor
%\end{eqnarray}
\end{lemma}
\begin{proof}
Let $u$ be a node originating a packet $packet_u$ and it has degree $d(u)$.
For any arbitrary node $v$, the packet $packet_u$ will be forwarded only if it is the first time to visit $v$ or the counter $c(x_u)\geq 2$.
We know that every packet originated from a node $u$ has a counter given by
\begin{eqnarray}
c(x_u)=\lfloor n/d(u)\rfloor.
\end{eqnarray}
  Let $\mu$ be the mean degree of an abstract graph representing the network $\N$, see Definition~\ref{eq:mu}. On average assuming every packet will be sent to $\mu$ neighboring nodes.
Approximating the mean degree of the graph to the degree of any arbitrary node $u$, the result follows.

\end{proof}

The previous lemma ensures that if $d(u) > n/2$, then the node $u$ will flood its packet only once $c(u)=1$.  In addition, nodes with smaller degrees will require to send their packets using large number of steps.

If the total number of nodes is not known, one can use a random walk initiated by the node $u$ to estimate the total number of nodes. In Section~\ref{sec:algDSA-II}  we will propose different algorithm that does not depend on estimating $n$ or use random walks in a graph.

The following lemma shows the total number of transmissions required to disseminate the information throughout the network.

\begin{lemma}\label{lem:totaltransmissions}
Let $\N$ be an instance model of a wireless sensor network with $n$ sensor nodes. The total number of transmissions required to disseminate the information from any arbitrary node throughout the network is given by
\begin{eqnarray}
O (n).
\end{eqnarray}
\end{lemma}
\begin{proof}
Let $d(s_i)$ be the degree (number of neighbors with a direct connection) of a sensor node $s_i$. On average $\mu$ is the mean degree of the set of sensors $S$ approximated to $\frac{1}{n}(\sum_{i}^n d(s_i))$. Every node does flooding that takes $O(1)$ running time to $d(s_i)$ neighbors. In order to disseminate information from a sensor $s_i$, at least $n/\mu$ steps are needed using Lemma~\ref{lem:onepacket}.  Also, every sensor $s_i$ needs to send $\mu$ messages on average to the neighbors. Hence the result follows.

\end{proof}

The following theorem shows the encoding complexity of DSA-I algorithm.
\begin{theorem}
The encoding operations of DSA-I algorithm are the total number of transmissions required to disseminate information sensed by all nodes that is given by
\begin{eqnarray}
O (n^2).
\end{eqnarray}
\end{theorem}

\section{DSA-II Algorithm Without Knowing Global Information}\label{sec:algDSA-II}
In   algorithm DSA-I we assumed that the total number of nodes are known in advance for each  sensing storing node in the network. This might not be the case since arbitrary  nodes might join and leave the network at various time due to the fact that they have limited CPU and short life time. Therefore, one needs to design network storage algorithm that does not depend on the value of total number of nodes.

In this section we will develop a distributed storage algorithm (DSA-II) that is totally distributed without knowing global information.  The objective is that  each node $u$ will estimate a value for its counter $c(u)$; the number of steps in which each packet will be disseminated in the network. In DSA-II each node $u$ will first perform an inference phase that will calculate  value of the counter $c(u)$. This can be achieved using the degree of $u$ and the degrees of the neighboring nodes $\N(u)$. We also assume a system parameter $c_u$ that will depend on the network condition and node's degree.

\bigskip

\noindent \textbf{Inference Phase:} Let $u$ be an arbitrary node in a distributed network $\N$. In the inference phase, each node $u$ will dynamically determine  value of the counter $c(u)$. The node $u$ knows its neighbors $\N(u)$. This is achieved in the flooding phase. Furthermore, the node $v$ in $\N(u)$ knows the degrees of  its neighbors.

The inference phase is done dynamically in a sense that every node in the network will separately decide a value for its counter. Nodes with large degrees will have a high chance of forwarding their data throughout the network to a large number of nodes.

Then encoding operations of DSA-II algorithm are similar to DSA-I algorithm except the former utilizes an inference phase, where the number of forwarding steps are predetermined first. Assume $v$ be a node connected to a source node $u$. Let $b_v$ be the degree of a node $v$ without adding  nodes in $\N(u) \cup u$. We can define the counter $c(u)$ as
\begin{eqnarray}
c(u)=c_u\Big\lfloor \frac{1}{d(u)} \sum_{v \in \N(u)} b_v \Big \rfloor
\end{eqnarray}

%
%%%%%%%%%%%%%%%%%%%%%%%%%%%%%%%%%%%%%%%%%%%%%%%%%%%%%%
\begin{algorithm}[t!]
\SetLine%
\KwIn{A sensor network $\N$ with $S=\{s_1,\ldots,s_i,\ldots\}$ source nodes,   source packets $x_{s_i},\ldots,x_{s_i},\ldots$.}%
\KwOut{storage buffers $y_1,y_2,\ldots,y_i,\ldots$ for all sensors $S$.} %
\ForEach{node $u$ in $\N$}%
    {
    %Generate $d_c(u)$ according to $\Omega_{is}(d)$ (or  %$\Omega_{rs}(d)$ and
    determine a set of neighbors $\N(u)$ using flooding.\;%
    determine a system parameter $c_u$\;
    }%
    Inference Phase\\
\ForEach{source node $u$ in $\N$}%
    {
    query the neighbors $\N(u)$ of $s_i$ for their degrees.\;
     Let $v\in \N(u)$ and $b_v$ be the v  degree without adding nodes in $\N(u)\cup u$\;
       \If{$d_v=1$}%
       {
       Repeat inference phase at $v$\;
       Repeat until $b_{v'} \neq 1$ for some $v' \in N(v)$\;
       Put $b_v=\sum_{v'}d_{v'}$}
       {}
    $c(u)=c_u\big\lfloor \frac{1}{d(u)}\sum_{v \in \N(u)} b_v \big\rfloor$\;%
    }
\ForEach{source node $s_i$ in $\N$}%
    {
    Generate header of $x_{s_i}$ and $token=0$\;%
       flood $x_{s_i}$ to all $\mathcal{N}(s_i)$ uniformly at random, send $x_{s_i}$ to $u \in \mathcal{N}(s_i)$ \;%

 with probability 1,  $y_u$ = $y_u\oplus x_{s_i}$\;%
 Put $x_{s_i}$ into $u$'s forward queue\;%
    $c(x_{s_i})=c(x_{s_i})-1$\;%
    }%
    \While{source packets remaining}%
{Run the \emph{encoding and flooding phase in DSA-I alg.}\;
}

\mbox{}\\%x

\caption{ DSA-II Algorithm: Distributed storage algorithm  for a WSN without knowing global information where
the data is disseminated using multicasting and flooding to all
neighbors.} \label{alg:DSA-I}
\end{algorithm}
%%%%%%%%%%%%%%%%%%%%%%%%%%%%%%%%%%%%%%%%%%%%%%%%

\bigskip

\noindent \textbf{Encoding and Flooding Phase:}
\begin{itemize}
\item After the inference and initialization phases, every node $u$ receiving the
    $packet_{s_i}$ will accept the data $x_{s_i}$ with probability
    one and will add this data to its buffer $y$.
\begin{eqnarray}
y_u^+=y_u^- \oplus x_{s_i}.
\end{eqnarray}
\item The node $u$ will decrease the counter by one as
    \begin{eqnarray}c(x_{s_i})=c(x_{s_i})-1.\end{eqnarray}
\item The node $u$ will select a set of neighbors that did not
    receiver the message $x_{s_i}$ and it will send this message
    using multicasting.
\item For an arbitrary node $v$ that receives the message from
    $u$, it will check if the $x_{s_i}$ has been received before,
    if yes, then it will discard it. If not, then it will run a
    probability distributed whether to accept or reject it. If
    accepted, then it will add the data to its buffer $
    y_v^+=y_v^- \oplus x_{s_i} $ and will decrease the counter
    $c(x_i)=c(x_i)-1.$

\item The node $v$ will check if the counter is zero, otherwise it will decrease it and send this message to the neighboring nodes that did not receive it.

\end{itemize}

\bigskip

\noindent  \textbf{Storage Phase:} Every node will maintain its own buffer
    by storing a copy of its data and other nodes' data. Also, a node
    will store a list of nodes ID's of the packets that reached it.
    After all nodes receive, send and storage their own and
    neighbors' data, every node will be able to maintain a buffer
    with some data of the network nodes.

\bigskip
%%%%%%%%%%%%%%%%%%%%%%%%%%%%%%%%%%%%%%%%%%%%%%%%%%%%%%%%%%%%%%%%%%%%%%%%
\section{DSA-II Analysis}\label{sec:analysis2}
We also shall provide analysis for the DSA-II algorithm shown in the previous section. The main idea is to utilize flooding and the node degree  to disseminate the sensed data from sensors throughout the network. We ensure that nodes with large degree will have  smaller counters in their packets such that their packets will travel for minimal number of hops. Also, nodes with smaller degree will have larger counters such that their packets will travel to many neighbors as possible.

The following lemma establishes  the number of hobs (steps) that every packet will travel in the network. Let $\lambda$ be the average node density~\cite{Pe03}.

\begin{lemma}\label{lem:onepacket2}
On average for a uniformly distributed network,  the total number of steps for one packet originated by a node $u$ in one branch in DSA-II is given by
\begin{eqnarray}
O(\mu-\lambda).
\end{eqnarray}

%\begin{eqnarray}
 %$c(u)=c_u\big\lfloor \sum_{v \in \N(u)} b_v /d_n(u)\big\rfloor$\;
%\end{eqnarray}
\end{lemma}
\begin{proof} %\NT{Revision needed}
Let $u$ be a node originating a packet $packet_u$ and it has degree $d(u)$ and when the nodes are uniformly distributed in the network we can approximate $d(u)$ as $\mu$.
We know that every packet originated from a node $u$ has a counter given by
\begin{eqnarray}
c(u)=c_u\Big\lfloor \frac{1}{d(u)}\sum_{v \in \N(u)} b_v\Big\rfloor.
\end{eqnarray}
We ensure that $c_u$ is inversely proportional to node degree so that nodes with small number of neighbors we take large values of $c_u$ and vice versa. Also in case that node $v$ has only one neighbor other than the originating node $u$ we traverse through this node until we get at least one node $v'$ that has degree $b_v' > 1$ .\\
   On average assuming every packet will be sent to $\mu$ neighboring nodes.
We can approximate $b_v$ as $\mu-\lambda$ so we can rewrite the equation  $\sum_{v \in \N(u)} b_v /d(u)$ as ${(\mu) (\mu-\lambda)} /\mu$.
For any arbitrary node $v$, the packet $packet_u$ will be forwarded only if it is the first time to visit $v$ or the counter $c(x_u)\geq 2$.
\end{proof}

The following lemma shows the total number of transmissions required to disseminate the information throughout the network.

\begin{lemma}\label{lem:totaltransmissions2}
Let $\N$ be an instance model of a wireless sensor network with $n$ sensor nodes uniformly distributed. The total number of transmissions required to disseminate the information from any arbitrary node throughout the network is given by
\begin{eqnarray}
O (\mu(\mu-\lambda)).
\end{eqnarray}
\end{lemma} %\NT{Revision needed}
\begin{proof}
Let $d(s_i)$ be the degree (number of neighbors with a direct connection) of a sensor node $s_i$. On average $\mu$ is the mean degree of the set of sensors $S$ approximated to $\frac{1}{n}(\sum_{i=1}^n d(s_i))$. Every node does flooding that takes $O(1)$ running time to $d(s_i)$ neighbors. In order to disseminate information from a sensor $s_i$, at least $\mu-\lambda$ steps are needed using Lemma~\ref{lem:onepacket2}.  Also, every sensor $s_i$ needs to send $\mu$ messages on average to the neighbors. Hence the result follows.
\end{proof}

The following theorem shows the encoding complexity of DSA-I algorithm.
\begin{theorem}
The encoding operations of DSA-II algorithm are the total number of transmissions required to disseminate information sensed by all nodes and given by
\begin{eqnarray}
O (\mu(\mu-\lambda) n).
\end{eqnarray}
\end{theorem}
%\newpage
\bigskip
%%%%%%%%%%%%%%%%%%%%%%%%%%%%%%%%%%%%%%%%%%%%%%%%%%%%%%%%%%%%%%%%%%%%%%%%
\section{Performance and Simulation Results}\label{sec:simulation}

In this section we will simulate the distributed storage algorithms, DSA-I and DSA-II,  presented in the previous sections. The main performance metric we investigate is the successful decoding probability versus the decoding ratio.

Let $\rho$ be the successful decoding probability defined as percentage of $M_s$  successful trials for recovering all $n$ variables (symbols) to the total number of trails. Also, let $h$ be the total number of queries needed to recover those $n$ variables. We can define the decoding ratio as the total queried nodes divided by n, i.e. $h/n$.

\vspace{+.1in}
\begin{definition}(Decoding Ratio)
\emph{Decoding ratio} $\eta$ is the ratio between the number of querying nodes
$h$ and the number of sources $n$, i.e.,
\begin{equation}\label{eq:eta}
\eta=\frac{h}{n}.
\end{equation}
\end{definition}
\vspace{+.1in}

\begin{definition}(Successful Decoding Probability)
\emph{Successful decoding probability} $P_s$ is the probability that the $n$
source packets are all recovered from the $h$ querying nodes.
\end{definition}
\vspace{+.1in}

\begin{figure}
\begin{center}
\includegraphics[scale=0.42]{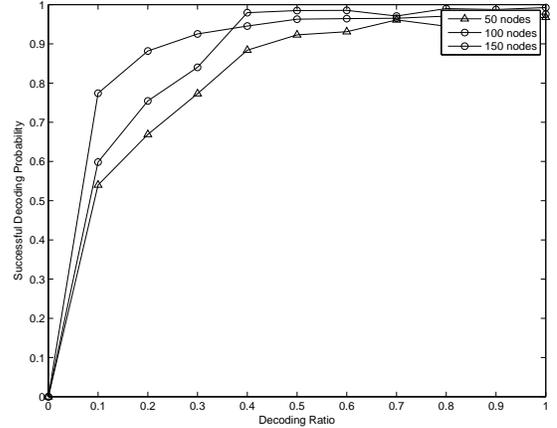}
\caption{A WSN with $n$ nodes arbitrary and randomly distributed in a field. The successful decoding ratio is shown for various values of n=50, 100, 150 with the DSA-I algorithm.}
\label{fig:DSA-I-2X2}
\end{center}
\end{figure}

In our simulation, $P_s$ is evaluated as follows. Suppose the network has n nodes , and we query h
nodes. There are $\binom{n}{h}$ ways to choose such h nodes, we pick a set $S$ of these choices uniformly at random, set $S$ was chosen large enough to give more normal results, So given the set $S$ which is a ratio $0<r\leq1$ of all possible combinations we define  $\mathcal{M}$  as fellow:

\begin{equation}\label{eq:M}
\mathcal{M} = r*\dbinom{n}{h} = r*\frac{n!}{h!(n-h)!}.
\end{equation}

Let $M_s$ be the size of the subset these $M$ choices of h query nodes from which the all $n$ source packets can be recovered. Then, we evaluate the successful decoding probability as

\begin{equation}\label{eq:Ps}
P_s =\frac{M_s}{\mathcal{M} }.
\end{equation}

We ran the experiment over a network with area $A = [L, L]^2$ grid and with different node densities. We evaluated the performance with various decoding ratios depending on the total number of nodes inside the network with incremental $\emph{step} = 0.1$.

For a decoding ratio $\eta$ we select $h$ nodes for our test. So we may have a large number of combinations to choose from, which may get order of $100^{100}$ combinations, So we have to choose a fair portion $r$ of these combinations $N \ll r \ll M$ and average the results over these experiments.

Fig.~\ref{fig:DSA-I-2X2} shows the decoding performance of DSA-I algorithm with Ideal Soliton distribution with small number of nodes.We ran the experiment over a network with area $A = [2, 2]^2$ grid and with a node density $2.5\leq \lambda \leq 12.5$.We evaluated the performance with various decoding ratio $0.1 \leq \eta \leq 1$ with incremental $\emph{step} = 0.1$.

From these results we can see that the successful decoding probability increases as the node density increases while  the decoding ratio $\eta$ is kept  constant.  We can deduce that the successful decoding probability is above $ \% 70 $  when the decoding ratio is about $\% 20--\% 30$. Another observation is that with a node density $\lambda > 8 $, the successful decoding probability $P_s > \% 90$.\\

\begin{figure}
\begin{center}
\includegraphics[scale=0.5]{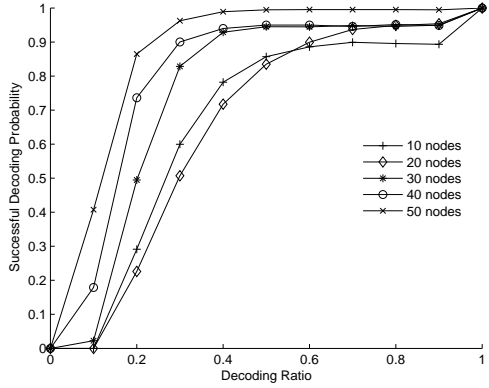}
\caption{A WSN with $n$ nodes arbitrary and randomly distributed in a field. The successful decoding ratio is shown for various values of n=30, 40, 50 with the DSA-II algorithm.}
\label{fig:DSAII2X2}
\end{center}
\end{figure}

Fig.~\ref{fig:DSA-I-5X5} shows the decoding performance of DSA-I algorithm with Ideal Soliton distribution with medium number of nodes. The network is deployed in $A = [5, 5]^2$ with node density $\lambda$ ranges from $4$ to $20$. From the simulation results we can see that the decoding ratio increases with the increase of $\lambda$ and approaches to 1 for $\eta > \%20$ and $\lambda \geq 12$.\\

Fig.~\ref{fig:DSAII2X2} shows the decoding performance of DSA-II algorithm with Ideal Soliton distribution with small number of nodes. We ran the first experiment over a network with area $A = [2, 2]^2$ grid and with a node density $2.5\leq \lambda \leq 12.5$, and evaluated the performance with various decoding ratio $0.1 \leq \eta \leq 1$ with incremental $\emph{step} = 0.1$, As shown in the figure the DSA-II algorithm archived similar results to the DSA-I algorithm with a successful decoding probability $P_s > \% 70 $ for a decoding ratio $\eta \geq 0.4 $.

\begin{figure}[t]
\begin{center}
\includegraphics[scale=0.4]{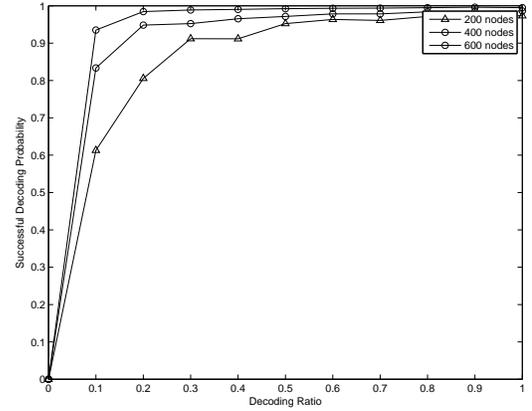}
\caption{A WSN with $n$ nodes arbitrary and randomly distributed in a field. The successful decoding ratio is shown for various values of n= 200, 400, 600 with the DSA-I algorithm.}
\label{fig:DSA-I-5X5}
\end{center}
\end{figure}

Fig.~\ref{fig:avgbuf5x5} shows the a caparison between the buffer size in DSA-I and DSA-II in a network deployed in an area $A = [5, 5]^2$, it can be concluded from the results that the buffer size approximately equals $\%10 $ of the network size $n$. From  Fig.~\ref{fig:avgbuf5x5} it can be seen that the buffer size is strongly related to the network density $\lambda $.

\bigskip

\section{Evaluation and Practical Aspects}
In this section we shall provide evaluation and comparison analysis between DSA-I and DSA-II algorithms and related work in distributed storage algorithms. Previous work focused on utilizing random walks and Fountain codes to disseminate data sensed by a set of sensors throughout the network. Also, global and geographical information such as knowing total number of nodes, routing tables, and node locations are used. In this work we do not assume knowing such global information.

The main goal of this work is to design data collection algorithms that can be utilized in large-scale wireless sensor networks. We achieve this goal by disseminate data throughout the network using data flooding once at every sensor node, then adding some redundancy at other neighboring nodes using random walks and packet trapping. Every storage node will keep track of other node's ID's, from which it will accept/reject packets.

The main advantages of the proposed algorithms are as follows
\begin{compactenum}[i)]
\item
One does not need to query all nodes in the network in order to retrieve information about all n nodes. Only $\%20-\%30 $ of the total nodes can be queried.
\item One can query only one arbitrary node $u$ in a certain region in the network to obtain an information about this region.

\end{compactenum}

\begin{figure}
\begin{center}
\includegraphics[scale=0.55]{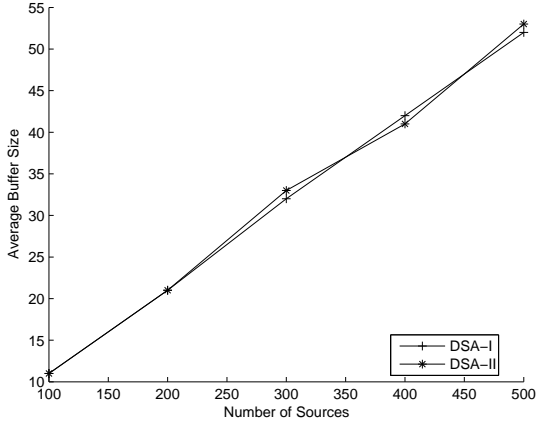}
\caption{A Caparison between DSA-I and DSA-II buffer size for various node densities in a medium size network. Increasing number of sensor nodes increases linearly the number of buffers.}
\label{fig:avgbuf5x5}
\end{center}
\end{figure}

\subsection{Sensing New Data}
The proposed algorithms work also in the case of data update. Assume a node $u$ sensed data $x_u$ and it has been disseminated throughout the network using flooding as shown in DSA-I and DSA-II algorithms. In this case the flag value is set to zero; and a packet from the node $u$ is originated as follows:

\begin{eqnarray}
packet_{u}=(ID_{u},x_{u},c(x_u), flag)
\end{eqnarray}
We notice that every node $v$ stores a copy from this data $x_u$ will also maintain a list of ID's including $ID_u$.

Assume $x_u^{'}$ be the new sensed data from the node $u$. Let us consider the case that the node $u$ wants to update its values, then the node $u$ will send update message setting the flag to one.
\begin{eqnarray}
packet_{u}=(ID_{u} ,x_u^{'}\oplus x_{u},c(x_u), flag).
\end{eqnarray}
The new and old data are Xored in this packet.
Every storage node will check the flag, whether it is an update or initial packet. Also, the node $v$ will check if $ID_u$ is in its own list. Once a node $v$ accepts the coming update packet, it will update its target buffer as
\begin{eqnarray}y_v^+=y_v^- \oplus x_u^{'}\oplus x_{u}.\end{eqnarray}
.

\subsection{Practical Aspects}
The proposed algorithms can be deployed in  large-scale wireless sensor networks, where geographic locations of sensor nodes are not known. Also, each sensor does not need to maintain routing tables about the neighboring nodes. Such applications include WSN's disseminated in forests and burned fields, where monitoring and detecting fires, floods and disasters phenomena  are required. It also can be deployed in crowd large fields, where a large number of nodes is scattered to collection data.

The proposed data collection and storage algorithms certainly are can be deployed in Minna and Arafat fields in the east south of Makkah during pilgrimage. Fig.~\ref{fig:minna2} shows camp tents located in Minna field in east of Makkah. The tents are supported by air-condition, electricity, and gas suppliers. The sensor devices are distributed randomly to measure gas pollution, detect fires, collect data, learn about the environment. The data storage devices receive collected  data by the sensors and send it to the main server for further analysis. More details and practical aspects of this model will be explained in our future work.

\begin{figure*}[t]
\begin{center}
\includegraphics[width=11cm, height=7cm]{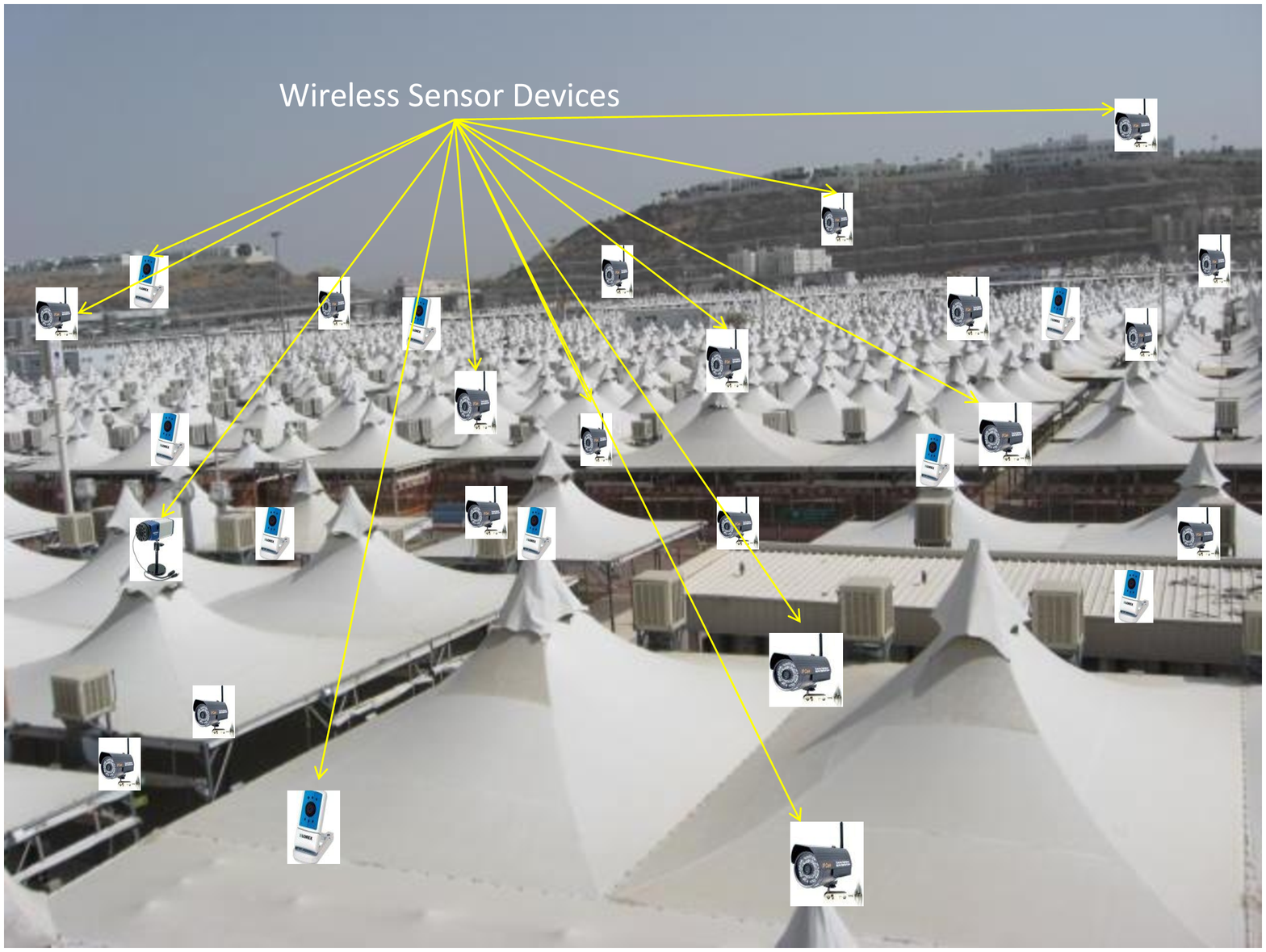}
\includegraphics[width=7cm, height=7cm]{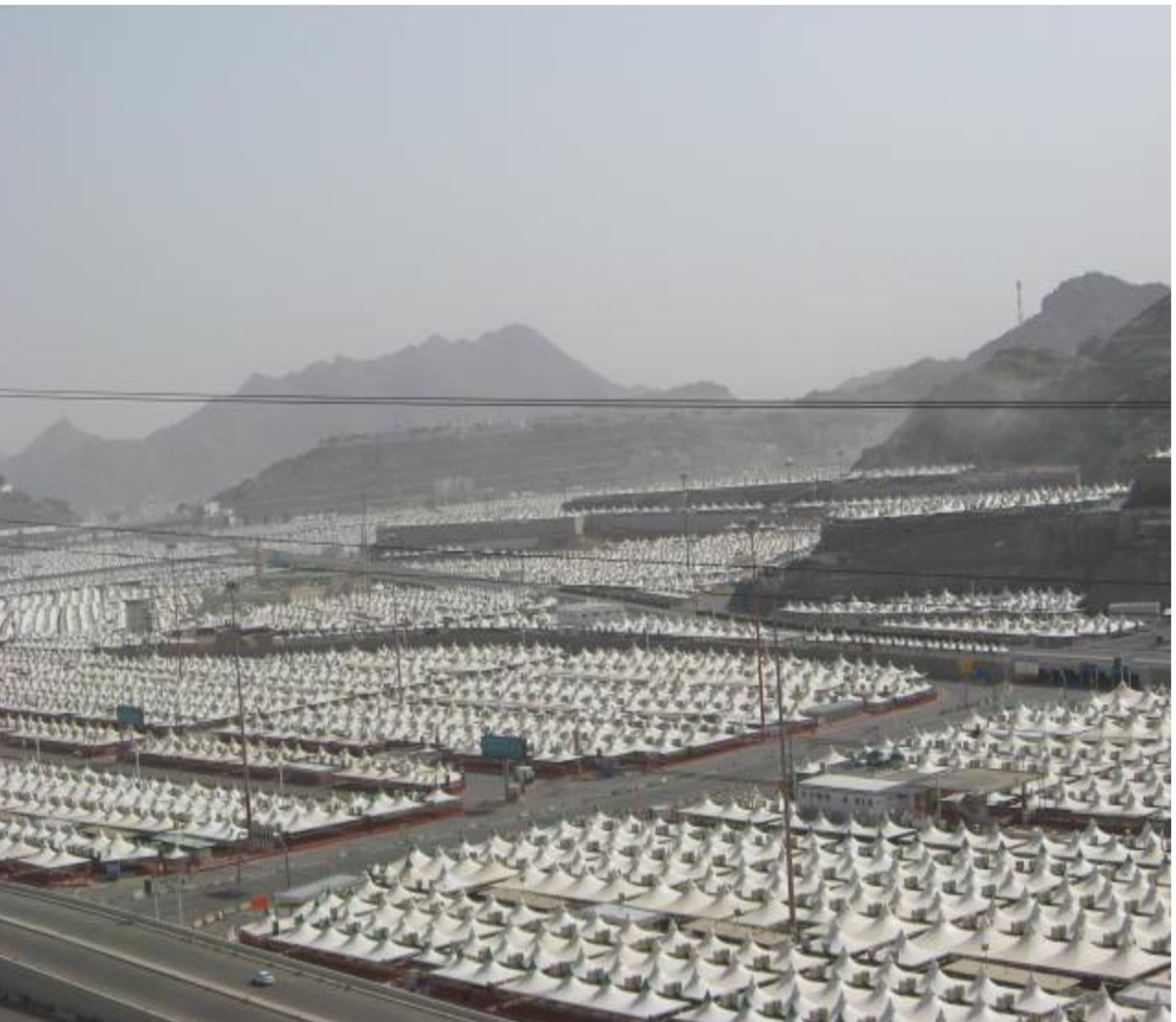}
\caption{Wireless sensor devices are scattered in Minna field in East of Makkah to gather and collect data about the environment. Such sensors are able to detect fires, gas pollution, and other disasters phenomena. They are needed to monitor the large number of camp tents in Minna field.}
\label{fig:minna2}
\end{center}
\end{figure*}

%%%%%%%%%%%%%%%%%%%%%%%%%%%%%%%%%%%%%%%%%%%%%%%%%%%%%%%%%%%%%%%%%%%%%%%%
\section{Related Work} \label{sec:relatedwork}
Wireless vision sensor networks are small devices that can be  scattered in a field or deployed in a network to measure certain phenomena.
In this section we present the previous work in network storage codes that is relevant to our work.  Distributed network storage codes such as Fountain codes are used along with random walks to distribute data from a set of sources $k$ to a set of storage nodes $n \gg k$, see~\cite{dimakis06b,aly08e}.
However, in this work we generalize this scenario where  a set of $n$ sources  disseminate their data to a set of $n$ storage nodes .

The most notable work in distributed storage algorithms for wireless sensor networks can be stated as.

\smallskip

\begin{itemize}
\item

Dimakis~\emph{el al.} in~\cite{dimakis06a,dimakis06b,dimakis08} used a decentralized
implementation of Fountain codes that uses geographic routing and every node
has to know its location. The motivation for using Fountain codes instead of
using random linear codes is that Fountain codes need $O(k \log k)$ decoding
complexity but random linear codes and RS codes use $O(k^3)$ decoding
complexity where k is the number of data blocks to be encoded. Also, one does
not know in advance the degree $d$ of the collector nodes~\cite{lin07b}. The authors propose
a randomized algorithm that constructs Fountain codes over grid network using
only geographical knowledge of nodes and local randomized decisions. They also
used fast random walks to disseminate source data to the storage nodes.
\smallskip

\item  Lin \emph{el al.} in~\cite{lin07a,lin07b} studied the
question "how to retrieve historical data that the sensors have gathered even
if some sensors are destroyed or disappeared from the network?"  They analyzed
techniques to increase "persistence" of sensed data in a random wireless sensor
network. They proposed two decentralized algorithms using Fountain codes to
guarantee the persistence and reliability of cached data on unreliable sensors.
They used random walks to disseminate data from a sensor (source) node to a set
of other storage nodes. The first algorithm introduces lower overhead than
naive random-walk, while the second algorithm has lower level of fault
tolerance than the original centralized Fountain code, but consumes much lower
dissemination cost.  They proposed the first novel decentralized implementation
of Fountain codes in sensor networks in an efficient and scalable fashion. The
authors did not use routing tables to dissimilate data from one sensor to a set
of sensors. The reason is that a sensor does not have enough energy or memory
to maintain a routing table which is scalable with the size of the network.

\item
Kamara \emph{el al.} in~\cite{kamra06} proposed a novel technique called
\emph{growth codes} to increase data persistence in wireless sensor networks,
i.e. increasing the amount of information that can be recover at the sink.
\emph{Growth codes} is a linear technique that information is encoded in an
online distributed way with increasing degree. They defined persistence of  a
sensor network as \emph{"the fraction of data generated within the network that
eventually reaches the sink"}~\cite{kamra06}. They showed that \emph{growth
codes} can increase the amount of information that can be recovered at any
storage node at any time period whenever there is a failure in some other
nodes.  They do not use robust or Soliton distributions, however, they propose
a new distribution depending on the network condition to determine degrees of
the storage nodes. The motivation for their work is that
\begin{inparaenum} \item Positions of the nodes are not known, so a sensor node
does not need to know positions of other nodes.  \item They assume a round time
of update the nodes, meaning with increasing the time $t$, degree of a symbol
is increased. This is the idea behind growth degrees. \item They provide
practical implementations of growth codes and compare its performance with
other codes. \item The decoding part is done by querying an arbitrary sink, if
the original sensed data has been collected correctly then finish, otherwise
query another sink node.
\end{inparaenum}

\medskip

\item
The authors in~\emph{el al.} in~\cite{aly08e,aly11a} studied a model for  distributed network storage algorithms for wireless sensor networks where $k$ sensor nodes (sources) want to disseminate their data to $n$ storage nodes with less computational complexity. The authors used Fountain codes and random walks in graphs to solve this problem. They also assumed that the total number or sources and storage nodes are not known. In  other words, they gave an algorithm where every node in a network can estimate the number of sources and the total number of nodes.

\end{itemize}

 In this work we propose a different system for a wireless sensor network where all nodes act as sources as well as storage/receiver nodes. The encoding operations of a node to disseminate its data are linear and take less computational time in comparison to the previous work.

%%%%%%%%%%%%%%%%%%%%%%%%%%%%%%%%%%%%%%%%%%%%%%%%%%%%%%%%%%%%%%%%%%%%%%%%
\section{Conclusion}\label{sec:conclusion}
In this work we presented two distributed storage algorithms for large-scale wireless sensor networks.
Given n storage nodes with limited buffers we demonstrated schemes to disseminate sensed data throughout the network with less computational overhead. The results and performance show that it is required to query only $\%20-\%30$ of the network nodes in order to retrieve the data collected by the n sensing nodes, when the buffer size is $\%10$ of the network size. Our future work will include practical and implementation aspects of these algorithms to better serve American-made camp tents in Minna and Arafat fields located in the east south of Makkah, KSA.

%\scriptsize
\bibliographystyle{plain}
%\bibliography{CodingStoragerefs}
%\bibliographystyle{ieeetr}

\section*{Appendix}
\bigskip

Given a network $\N$, the mean degree of a node in $G$ can be defined as:
\goodbreak
\begin{definition}(Node Degree)
Consider a graph $G=(V,E)$, where $V$ and $E$ denote the set of nodes and
links, respectively. Given $u,v\in V$, we say $u$ and $v$ are \emph{adjacent}
(or $u$ is adjacent to $v$, and vice versa) if there exists a link between $u$
and $v$, i.e., $(u,v)\in E$. In this case, we also say that $u$ and $v$ are
\emph{neighbors}. Denote by $\mathcal{N}(u)$ the set of neighbors of a node
$u$. The number of neighbors, with a direct connection, of a node $u$ is called the \emph{node degree} of
$u$, and denoted by $d(u)$, i.e., $|\mathcal{N}(u)|=d(u)$. The \emph{mean
degree} of a graph $G$ is  given by
\begin{equation}\label{eq:mu}
\mu = \frac{1}{|V|}\sum_{u\in G}d(u),
\end{equation}
where $|V|$ is the total number of nodes in $G$.
\end{definition}

The
Ideal Soliton distribution $\Omega_{is}(d)$ for $k$ source blocks is given by
\begin{equation}\label{eq:Ideal-Soliton-distribution}
\Omega_{is}(i)=\Pr(d=i)=\left\{ \begin{array}{ll} \vspace{+.05in} \displaystyle \frac{1}{k}, & i=1\\
\displaystyle\frac{1}{i(i-1)}, & i=2,3,...,k.\end{array}\right.
\end{equation}

Let $R=c_0\ln(k/\delta)\sqrt{k}$, where $c_0$ is a suitable constant and
$0<\delta<1$. The Robust Soliton distribution for $k$ source blocks is defined
as follows. Define
\begin{equation}
\tau(i)=\left\{ \begin{array}{ll} \vspace{+.05in} \displaystyle \frac{R}{ik}, & i=1,...,\displaystyle\frac{k}{R}-1\\
\vspace{+.05in} \displaystyle \frac{R\ln(R/\delta)}{k}, & i=\displaystyle\frac{k}{R}, \\
0, & i=\displaystyle\frac{k}{R}+1,...,k,\end{array}\right.
\end{equation}
and let
\begin{equation}
\beta=\sum_{i=1}^k \tau(i)+\Omega_{is}(i).
\end{equation}
The Robust Soliton distribution is given by
\begin{equation}\label{eq:Robust-Soliton-distribution}
\Omega_{rs}(i)=\frac{\tau(i)+\Omega_{is}(i)}{\beta}, \mbox{ for all }
i=1,2,...,k
\end{equation}

\end{document}